\pdfoutput=1
\documentclass[lettersize,journal]{IEEEtran}
\usepackage{amsmath,amsfonts}
\usepackage{array}
\usepackage[caption=false,font=normalsize,labelfont=sf,textfont=sf]{subfig}
\usepackage{textcomp}
\usepackage{stfloats}
\usepackage{url}
\usepackage{verbatim}
\usepackage[pdftex]{graphicx}
\DeclareGraphicsExtensions{.pdf,.png,.jpg}
\usepackage{cite}
\usepackage{braket}
\usepackage{algorithm}
\usepackage{algpseudocode}
\usepackage{booktabs}
\usepackage{siunitx}
\usepackage{amssymb, amsthm}
\usepackage{amsmath}

\newtheorem{theorem}{Theorem}

\newtheorem{proposition}{Proposition}[section]

\begin{document}

\title{Q‑EnergyDEX: A Zero‑Trust Distributed Energy Trading Framework Driven by Quantum Key Distribution and Blockchain}

\author{Ziqing Zhu,~\IEEEmembership{Member,~IEEE} 
}
\maketitle
\bstctlcite{BSTcontrol}

\begin{abstract}
The rapid decentralization and digitalization of local electricity markets have introduced new cyber-physical vulnerabilities, including key leakage, data tampering, and identity spoofing. Existing blockchain-based solutions provide transparency and traceability but still depend on classical cryptographic primitives that are vulnerable to quantum attacks. To address these challenges, this paper proposes Q-EnergyDEX, a zero-trust distributed energy trading framework driven by quantum key distribution and blockchain. The framework integrates physical-layer quantum randomness with market-level operations, providing an end-to-end quantum-secured infrastructure. A cloud-based Quantum Key Management Service continuously generates verifiable entropy and regulates key generation through a rate-adaptive algorithm to sustain high-quality randomness. A symmetric authentication protocol (Q-SAH) establishes secure and low-latency sessions, while the quantum-aided consensus mechanism (PoR-Lite) achieves probabilistic ledger finality within a few seconds. Furthermore, a Stackelberg-constrained bilateral auction couples market clearing with entropy availability, ensuring both economic efficiency and cryptographic security. Simulation results show that Q-EnergyDEX maintains robust key stability and near-optimal social welfare, demonstrating its feasibility for large-scale decentralized energy markets.\end{abstract}

\section{Introduction}
\subsection{General Background}
\IEEEPARstart{T}{he} ongoing decentralization and digitalization of local electricity markets introduces new cyber-physical vulnerabilities. These include key leakage, data tampering, and identity spoofing, particularly in peer-to-peer trading and real-time dispatch environments \cite{kaur2021tits_blockchain_cps}. In response, blockchain technology has emerged as a promising tool for ensuring transparency, traceability, and auditability in energy transactions \cite{kaur2021tits_blockchain_cps}. However, most blockchain-based systems rely on classical cryptographic primitives, such as RSA or elliptic-curve cryptography, whose security is contingent on computational hardness assumptions. The advent of quantum computing threatens to break these assumptions, exposing energy markets to substantial security risks \cite{moody2022ieeesp_pqc_standards}. To mitigate these risks, quantum key distribution (QKD) offers a fundamentally different security guarantee, grounded in the laws of quantum mechanics. It provides information-theoretic confidentiality, regardless of an adversary’s computational power \cite{wang2022natphotonics_tfqkd_830km, liu2023prl_tfqkd_1000km}. Although QKD has demonstrated strong potential in securing communication links at the physical and network layers, including simplified long-haul implementations without service-fibre dissemination \cite{zhou2023natcomm_tfqkd_no_service_fiber}, its integration with electricity markets remains largely unexplored. Existing QKD deployments in utility environments primarily serve as isolated key exchange systems, without direct interaction with operational market components \cite{evans2021ieeeaccess_utility_qkd}. In critical market functions such as identity authentication, consensus agreement, and encrypted bidding, the use of quantum entropy is still ad hoc or entirely absent \cite{evans2021ieeeaccess_utility_qkd}. Bridging this gap requires rethinking how QKD-generated randomness can be systematically integrated into the core mechanisms of energy markets.

\subsection{Literature Review}
Recent research at the intersection of power markets, blockchain systems, and quantum security has evolved along three primary threads. The first thread focuses on blockchain-based platforms for transaction and settlement within the power sector. Numerous studies have explored mechanisms to enable decentralized energy trading and multi-party settlement via smart contracts and distributed ledgers; representative designs include privacy-preserving transactive energy markets and peer-to-peer trading schemes with on-chain settlement \cite{Saha2021AppliedEnergy,Kumari2022Energies}. However, these designs generally assume trust in the underlying consensus mechanism and rely on traditional cryptographic primitives, limiting their robustness in adversarial environments. Moreover, the unpredictability of user behavior and limited modeling of system dynamics hinder their scalability and operational realism.

The second thread centers on improving consensus mechanisms and public randomness generators for blockchain infrastructures. A line of work has explored verifiable random functions, distributed key generation, and low-latency Byzantine fault-tolerant protocols to address the high-latency and resource-intensive nature of proof-of-work systems \cite{MRV99,AlgorandSOSP2017,HotStuff2019}. Further developments in fast finality and BFT-style protocols such as Tendermint improve throughput and responsiveness, which are crucial for real-time energy dispatch \cite{Tendermint2018}. Nevertheless, most of these systems are computationally secure rather than information-theoretically secure. They are vulnerable to potential key leakage or manipulation by compromised committee members, and their ability to generate truly unpredictable and uncontrollable randomness remains limited.

The third thread investigates the applicability and robustness of quantum key distribution (QKD) and post-quantum cryptography (PQC) in critical infrastructure. Field trials and utility-side deployments have demonstrated the feasibility of QKD for securing grid communications and authenticating control traffic, while information-theoretic extraction bounds formalize how to obtain nearly uniform keys from noisy measurements \cite{Alshowkan2022SciRep,Evans2021Access,Tomamichel2011Leftover}. In parallel, the NIST PQC standardization effort provides quantum-resistant public-key primitives to complement or replace classical schemes \cite{NISTPQCproject},
\cite{NISTPQCEncryption}.

\subsection{Research Gaps and Contributions}
In summary, current literature falls short in several key areas. First, although many blockchain-based designs exist for power markets, they lack quantum-resilient primitives and do not guarantee forward secrecy or randomness integrity under strong adversarial settings. Second, while consensus protocols have improved in latency and decentralization, most remain computationally secure and fail to leverage hardware-based entropy sources. Third, despite rapid advancements in QKD deployment, there remains a gap in its integration with blockchain systems and its co-design with real-time market infrastructure.

\begin{figure*}[t]
    \centering
    \includegraphics[width=\textwidth]{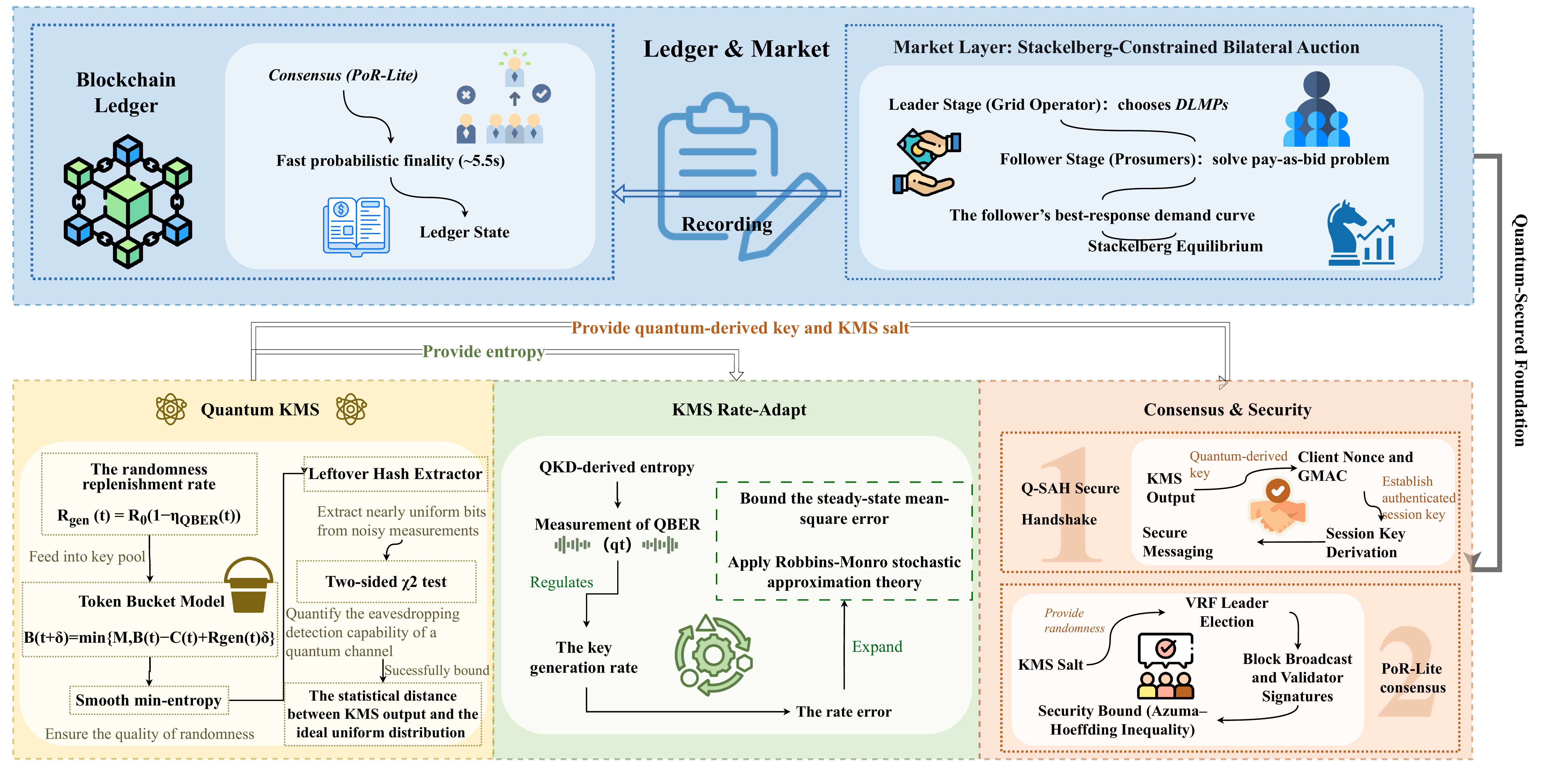}
    \caption{Framework}
    \label{fig:framework}
\end{figure*}

This paper proposes Q-EnergyDEX, a unified and quantum-resilient market infrastructure that integrates quantum key distribution (QKD), lightweight authentication, blockchain consensus, and market mechanisms tailored for power system applications. The main contributions are as follows:

1) We develop a lightweight symmetric handshake protocol (Q-SAH) that enables secure identity binding between devices using entropy generated from noisy QKD traces. This protocol eliminates reliance on traditional public key infrastructures and reduces cryptographic overhead during device registration and message authentication. Its design ensures compatibility with real-time constraints in electricity trading, while preserving information-theoretic security against quantum adversaries.

2) We propose PoR-Lite, a quantum-aided consensus mechanism that utilizes clipped entropy pools to achieve probabilistic finality with forward secrecy. The consensus protocol integrates verifiable randomness derived from quantum keys into a low-latency leader election process, maintaining ledger consistency and resistance to adversarial forks. Theoretical analysis verifies essential blockchain properties, including liveness and safety, under bounded randomness assumptions.

3) We design a quantized uniform-price auction mechanism that explicitly incorporates rate-adaptive entropy constraints into market clearing. By coupling the bidding process with secure randomness availability, the mechanism ensures that auction outputs remain within the allowable quantum entropy budget. The mechanism also preserves incentive compatibility and allocative efficiency under prosumer heterogeneity.

4) We build a comprehensive simulation environment that integrates quantum physical layer modeling, blockchain consensus dynamics, and electricity market operations. The environment captures the interaction between quantum channel conditions and market performance, enabling evaluation of the proposed architecture under diverse prosumer behaviors and varying quantum noise patterns. Empirical results confirm the security-efficiency tradeoffs and demonstrate the practical viability of the proposed system.

\section{System Model and Core Mechanisms}
Q‑EnergyDEX is designed as a two-layer architecture for secure and real-time energy trading. The upper layer ensures economic efficiency through market-based clearing, while the lower layer guarantees tamper-resistant consensus and randomness provisioning. To rigorously analyze this architecture, we first describe the participating entities and their interactions, then introduce the timing constraints imposed by smart grid operations. We next present the cloud-based Quantum Key Management Service (KMS), which leverages QKD-derived entropy to provide verifiable randomness for both the handshake and blockchain consensus protocols, thereby coupling physical-layer security with ledger consistency. The stochastic stability of the KMS key pool is then analyzed to ensure sustainable operation under high-frequency load. Finally, we show how these building blocks support two critical components of Q‑EnergyDEX: a Stackelberg-constrained bilateral auction for clearing and a PoR‑Lite consensus protocol that achieves fast probabilistic finality.

\subsection{Entity Sets and Roles}

To rigorously analyze Q‑EnergyDEX, we first clarify the system participants and their interactions before delving into timing constraints and key technical components. Let 
\(
\mathcal{P} = \{P_1, \dots, P_N\}
\) denote the set of \emph{prosumer nodes}, which can both purchase and sell energy; 
\(
\mathcal{K} = \{K_1, \dots, K_m\}
\) denote the set of cloud-based \emph{Quantum Key Management Service (KMS)} nodes that provide quantum-derived symmetric keys on demand and can be deployed in an active–active configuration for high availability; and 
\(
\mathcal{B} = \{B_1, \dots, B_N\}
\) denote the set of blockchain consensus nodes. For notational convenience we assume $\mathcal{B}$ is isomorphic to $\mathcal{P}$, but we still treat it as an independent set to emphasize modularity and allow alternative deployments. Finally, we define 
\(
\mathcal{A}
\) as the set of adversarial or Byzantine nodes that may attempt to eavesdrop, tamper with, or partially control the network. From a network topology perspective, the overall system can be seen as a two-layer structure: the upper layer is the blockchain P2P network, while the side edges connect each node to the cloud KMS API via virtual links, highlighting the decoupling between ledger consensus and the randomness provisioning service. Having identified the key entities and their layered topology, we next examine the timing and communication constraints that tightly couple these roles in a real-time smart grid environment.

\subsection{Communication and Timing Model}

Q-EnergyDEX must satisfy stringent latency and synchrony requirements at the communication and control layer, even though the real-time electricity market itself clears on a 5- to 15-minute trading window. To capture the fast sub-second dynamics that occur inside each trading cycle, we model the system as operating in discrete time slots of
\begin{align}
\Delta t = 100~\mathrm{ms}
\end{align}
This slotting does not imply that market settlements are performed every 100\,ms; rather, it provides a conservative abstraction for the fine-grained communication steps needed to complete secure handshakes and consensus rounds within each 5- or 15-minute trading window. In practice, under ideal or near-ideal network conditions, a single handshake or simple leader election can typically finish within 30-80 ms, making 100\,ms a reasonable and not overly optimistic assumption.

Each logical communication link between two nodes is modeled with a round-trip time (RTT) consisting of a deterministic baseline propagation delay $d_0$ and a bounded jitter term $\varepsilon_{ij}$:
\begin{align}
\mathrm{RTT}_{ij} &= d_0 + \varepsilon_{ij}, & 0 &\le \varepsilon_{ij} \le \varepsilon_{\max} = 15~\mathrm{ms}
\end{align}
When analyzing the responsiveness of the PoR-Lite consensus and Q-SAH handshake protocols, the expected time to finalize a single block or handshake exchange is therefore expressed as
\begin{align}
\mathbb{E}[T_{\mathrm{block}}] = \mathrm{RTT}_{ij}  + \mathcal{O}(1)
\end{align}
where the $\mathcal{O}(1)$ term accounts for lightweight cryptographic operations such as VRF evaluation and MAC verification, which remain negligible compared to network propagation delay.

\subsection{Quantum‑KMS Randomness Model}
The above stringent timing constraints motivate the need for a secure yet lightweight randomness provisioning service that can integrate seamlessly with both consensus and economic clearing protocols. To address this, we introduce the cloud-based Quantum Key Management Service (Quantum‑KMS, abbreviated as KMS) as an \emph{on-demand randomness pool} accessible through REST/gRPC APIs. Each prosumer node issues a \texttt{Rent(n)} request to obtain a one-time key of length $n$ and synchronously receives $(\texttt{key\_ID},\,K,\,\texttt{ttl})$, where $\texttt{ttl}$ denotes the maximum key lifetime (defaulting to 30\,s). Once expired, the key is retired via \texttt{Retire(key\_ID)}, while \texttt{Audit(sid,\,t)} returns the bit consumption statistics and anomaly alarms for a session $sid$ within a time window $t$. All interfaces are designed to respond within millisecond latency. To avoid duplicate key issuance, the same \texttt{key\_ID} across multiple KMS replicas is reconciled using a \emph{conflict-free replicated data type} (CRDT) merge rule
\begin{align}
\mathsf{merge}(x,\,y) = \min(x,\,y)
\end{align}
ensuring that no bit is issued twice. Network partitions may cause slight wastage but never compromise security.

The randomness replenishment rate of the KMS is denoted by $R_{\mathrm{gen}}(t)$, which evolves as
\begin{align}
R_{\mathrm{gen}}(t) = R_0 \bigl(1 - \eta_{\mathrm{QBER}}(t)\bigr) \quad [\mathrm{bit/s}]
\end{align}
where $R_0$ is the theoretical photon emission rate and $\eta_{\mathrm{QBER}}(t)$ reflects the reduction induced by the time-varying quantum bit error rate (QBER). Internally, the KMS can be abstracted as a \emph{token bucket}: let $B(t)$ denote the available bit balance at time $t$. Its dynamics follow
\begin{align}
B(t+\delta) = \min\!\Bigl\{M,\,B(t)-C(t)+R_{\mathrm{gen}}(t)\,\delta\Bigr\}
\end{align}
where $M$ is the bucket capacity and $C(t)$ is the actual consumption rate of the clients. A rental request fails if $n > B(t)$. For multi-region deployments with $m$ active replicas, requests are load-balanced via ECMP, and eventual consistency is guaranteed by asynchronous replication.

To ensure the quality of randomness, the physical measurement output is modeled as a random variable $X \in \{0,1\}^n$ accompanied by an error count $E \sim \mathrm{Binom}(n,\,q)$, where $q$ denotes the QBER. We quantify the security using the \emph{smooth min-entropy}. The $\varepsilon$-smooth min-entropy is defined as
\begin{align}
H_\infty^\varepsilon(X \mid E) 
  &= 
  \max_{\substack{\tilde X :\\ \operatorname{SD}(X,\tilde X)\le\varepsilon}}
  \min_{e} 
  \Biggl[
    -\log_2 
    \max_{x}\,
    \Pr\!\Bigl(
        \tilde X = x 
        \,\Big|\, 
        E = e
    \Bigr)
  \Biggr]
\end{align}
The $\varepsilon$-smooth min-entropy quantifies the residual unpredictability of $X$ against an adversary observing $E$. The term $\max_{x}\Pr(\tilde X=x \mid E=e)$ represents the adversary’s best-guess probability under a given error count, and the negative logarithm converts this into an uncertainty measure. The outer $\min_{e}$ enforces security in the worst-case error realization. To prevent rare outliers from unduly reducing the entropy, smoothing is introduced: we maximize over auxiliary distributions $\tilde X$ that are within statistical distance $\varepsilon$ of the true $X$, i.e., $\operatorname{SD}(X,\tilde X)\le\varepsilon$, thereby discounting low-probability anomalies. This yields a conservative yet practical measure of the extractable secure randomness.

Under a binary symmetric channel (BSC) model, the following entropy bound holds:
\begin{align}
H_\infty^\varepsilon(X \mid E) \ge n\bigl(1 - h_2(q)\bigr) - \log_2\!\frac{1}{\varepsilon}
\end{align}
where $h_2(\cdot)$ is the binary entropy function. This result can be derived using the Azuma–McDiarmid concentration inequality and shows that a lower QBER leads to extracted randomness closer to an ideal uniform distribution.

To extract nearly uniform bits from noisy measurements, the KMS applies a \emph{Leftover Hash Extractor}. Given an $n$-bit raw key, the length of secure key bits $\ell$ after privacy amplification is
\begin{align}
\ell = n\bigl(1 - h_2(q)\bigr) - 2\log_2\!\frac{1}{\varepsilon} - 128.
\end{align}
For a typical QBER of $q=0.02$ (2\%) and $\varepsilon=2^{-64}$, we obtain $\ell \approx n - 13$, meaning that from each 256-bit raw key we can still extract approximately 243 bits that are statistically close to uniform.

The ability of the quantum channel to detect eavesdropping can be quantified by a statistical hypothesis test. Let the physical baseline QBER be $q_0 \approx 1\%$, and suppose an adversary introduces an additional error $\Delta q$. In an $n$-bit sampling window, the miss-detection probability of a two-sided $\chi^2$ test is
\begin{align}
P_{\mathrm{miss}} = Q_{\chi^2}\!\Bigl( \frac{(\Delta q)^2\,n}{q_0} \Bigr),
\end{align}
where $Q_{\chi^2}$ is the complementary cumulative distribution function of the $\chi^2$ statistic. By requiring $P_{\mathrm{miss}}\le10^{-6}$, we find that with $n=10^6$ bits only an additional $\Delta q \ge 0.2\%$ is sufficient to trigger an alarm. This means that the cloud KMS can detect even subtle passive eavesdropping or channel tampering within an extremely short time window, satisfying the “unreliable-but-detectable” channel assumption of QKD.

Combining the smooth entropy bound with the detection probability, we further bound the statistical distance between KMS output and the ideal uniform distribution $U_n$. Let the ideal functionality $\mathcal{F}_{\mathrm{RND}}$ always output a perfectly uniform $n$-bit string. Then we have
\begin{align}
\operatorname{SD}\!\bigl[\mathsf{Real}^{X \mid E},\; \mathsf{Ideal}^{U_n}\bigr]
\;\le\; 2^{-64} + 10^{-6} \;<\; 2^{-50}.
\end{align}
Even an adversary with polynomial-time quantum circuits cannot distinguish the real KMS output from the ideal uniform source with advantage exceeding $2^{-50}$. In other words, the KMS output can be treated as equivalent to the ideal randomness source $\mathcal{F}_{\mathrm{RND}}$ under the UC-Hybrid framework, enabling the Q‑SAH symmetric handshake protocol and the PoR‑Lite consensus to be composed securely atop this ideal functionality.

\subsection{Key Pool Birth--Death Chain Model}

While the KMS model ensures the cryptographic quality of randomness, it must also remain stable under bursty and high-frequency transaction loads. To ensure sustainable operation, we analyze the key pool dynamics using a Birth–Death Markov chain. The state space is $\mathcal{S} = \{0,1,\dots,M\}$, where $s$ denotes the number of available bits currently stored in the key pool. The stochastic process $S(t)$ evolves over time under two independent drivers: a \emph{birth process} corresponding to the replenishment of random bits from the physical quantum source at rate $\mu$, and a \emph{death process} corresponding to random key consumption triggered by client transactions, which arrive as a Poisson process of rate $\lambda$ and consume $k$ bits per transaction, yielding an effective consumption rate of $\lambda k$.  

In this model, the \emph{birth rate} is given by $\beta_s = \mu$ when $s<M$ (no births occur at the full capacity $M$), while the \emph{death rate} is given by $\delta_s = \lambda k$ when $s>0$ (no further deaths are possible at the empty state $0$). Hence, the transition probabilities in an infinitesimal interval $\Delta$ satisfy:
\begin{align}
\Pr\!\bigl\{S(t+\Delta) = s+1 \,\big|\, S(t)=s\bigr\} &= \mu\,\Delta + o(\Delta), \quad s<M, \\
\Pr\!\bigl\{S(t+\Delta) = s-1 \,\big|\, S(t)=s\bigr\} &= \lambda k\,\Delta + o(\Delta), \quad s>0, \\
\Pr\!\Bigl\{
  S(t+\Delta) = s \,\Big|\, S(t)=s
\Bigr\} 
&= 1 - (\mu + \lambda k)\Delta \notag\\
&\quad + o(\Delta), \quad 0<s<M.
\end{align}

We now derive a closed-form expression for the empty-pool probability and the stability condition of this chain.

\begin{theorem}[Steady-State Distribution and Empty-Pool Probability]
\label{thm:keypool}
Consider a Birth--Death chain with finite capacity $M$, birth rate $\mu$ and death rate $\lambda k$, and let
\begin{align}
\rho = \frac{\lambda k}{\mu}.
\end{align}
If $\rho < 1$ (i.e., the replenishment rate exceeds the consumption rate), then the Markov chain is irreducible and positive recurrent, and the unique steady-state distribution is
\begin{align}
\pi_s = \frac{(1-\rho)\,\rho^s}{1 - \rho^{M+1}}, \quad s \in \{0,1,\dots,M\}.
\end{align}
As the capacity $M \to \infty$, the distribution converges to a geometric form $\pi_s = (1-\rho)\rho^s$, yielding the empty-pool probability
\begin{align}
\pi_0 = 1 - \rho.
\end{align}
Conversely, when $\rho \ge 1$, the chain becomes null recurrent or transient, and the key pool will almost surely be depleted in the long run, i.e., $\pi_0 \to 1$.
\end{theorem}

\begin{proof}
We begin with the global balance equations. Let $\pi_s$ denote the stationary probability of state $s$. For $0 < s < M$, the detailed balance yields
\begin{align}
\pi_{s-1} \mu + \pi_{s+1} \lambda k &= \pi_s (\mu + \lambda k).
\end{align}
At the boundaries we have
\begin{align}
\pi_0 \mu &= \pi_1 \lambda k, \qquad \pi_{M-1} \mu = \pi_M \lambda k.
\end{align}
Introducing the birth–death ratio $r = \mu/(\lambda k) = 1/\rho$, the local balance condition simplifies to
\begin{align}
\pi_s \lambda k &= \pi_{s-1} \mu, \quad 1 \le s \le M.
\end{align}
This recursion gives
\begin{align}
\pi_s &= \pi_0 \Bigl(\frac{\mu}{\lambda k}\Bigr)^s = \pi_0\, r^s = \pi_0\,\rho^{-s}.
\end{align}
To normalize for finite capacity $M$, we impose
\begin{align}
\sum_{s=0}^M \pi_s = \pi_0 \sum_{s=0}^M r^s &= \pi_0\,\frac{1 - r^{M+1}}{1 - r} = 1.
\end{align}
Therefore,
\begin{align}
\pi_0 = \frac{1 - r}{1 - r^{M+1}} = \frac{1 - \rho}{1 - \rho^{M+1}},
\end{align}
and
\begin{align}
\pi_s = \frac{(1-\rho)\,\rho^s}{1 - \rho^{M+1}}.
\end{align}
When $M \to \infty$ and $\rho<1$, the denominator $1-\rho^{M+1}\to 1$, yielding the simplified geometric distribution $\pi_s = (1-\rho)\rho^s$ and thus $\pi_0=1-\rho$. If $\rho \ge 1$, the denominator vanishes, the chain loses positive recurrence, and the process will almost surely hit the empty state, $\pi_0 \to 1$. \qedhere
\end{proof}

Based on Theorem~\ref{thm:keypool}, we can further derive the minimal key pool capacity required to satisfy a design target on the empty-pool probability. For finite capacity $M$, the normalization yields
\begin{align}
\pi_0 = \frac{1-\rho}{1-\rho^{M+1}} \le 10^{-9}.
\end{align}
Rearranging terms gives
\begin{align}
1 - \rho^{M+1} \ge (1-\rho)\cdot 10^{9}.
\end{align}
When $\rho<1$ and $(1-\rho)\cdot 10^9 \ll 1$, we approximate
\begin{align}
\rho^{M+1} \le 1 - (1-\rho)\cdot10^{9} \approx \exp\!\bigl[-(1-\rho)\,10^{9}\bigr].
\end{align}
Taking logarithms yields the lower bound on the minimal capacity:
\begin{align}
M \ge \frac{\ln(10^9)}{\ln(1/\rho)} - 1.
\end{align}
This closed-form bound explicitly relates the required pool capacity $M$ to the consumption and replenishment rates $\lambda k$ and $\mu$. Compared with a naive $M/M/1$ queue approximation, the above derivation provides a stricter convergence condition and a normalized factor for finite-capacity scenarios. 

Therefore, by formally modeling the key pool as a finite-capacity Birth--Death Markov chain, we rigorously characterize the stability of randomness under high-frequency trading loads and explicitly reveal the functional relationship among empty-pool probability, required capacity, and the replenishment/consumption rates. This analysis provides a mathematically grounded guideline for choosing KMS engineering parameters such as replenishment link rate and bucket capacity, which is later validated in Section~7.

\subsection{Stackelberg-Constrained Bilateral Auction}

Equipped with stable randomness and bounded latency, Q‑EnergyDEX can implement real-time market mechanisms. We first present a Stackelberg-constrained bilateral auction to achieve economically efficient energy clearing under network constraints. In the \emph{leader stage}, the grid operator chooses the vector of line-loss and congestion shadow prices $\mathbf u \ge 0$, i.e., the dynamic distributed locational marginal prices (DLMPs), to minimize its own dispatch cost while satisfying the physical network constraints:
\begin{align}
    \min_{\mathbf u \ge 0} \; C_{\mathrm{grid}}(\mathbf u)
    \quad \text{s.t.} \quad \mathbf H\!\bigl(\mathbf p^{\star}(\mathbf u)\bigr) \le \mathbf P^{\max},
\end{align}
where $\mathbf H \in \mathbb{R}^{B\times N}$ denotes the power transfer distribution factor (PTDF) matrix, $\mathbf P^{\max}$ is the vector of thermal limits, and $\mathbf p^{\star}(\mathbf u)$ is the aggregate optimal injection resulting from the followers’ best responses.

In the \emph{follower stage}, each prosumer node $P_i$ solves the following pay-as-bid problem:
\begin{align}
    \max_{p_i} \;\; \bigl(\pi_i - u_b H_{bi}\bigr)\, p_i - C_i(p_i),
    \quad |p_i| \le P_i^{\max},
\end{align}
where $\pi_i$ is the private valuation of node $i$, $u_b$ is the DLMP at bus $b$, and $C_i(\cdot)$ is a strictly convex quadratic cost function.

Writing the Karush–Kuhn–Tucker conditions for the follower problem yields the first-order optimality:
\begin{align}
    C_i'(p_i^\star) = \pi_i - u_b H_{bi}.
\end{align}
Therefore, the follower’s best-response demand curve is linear:
\begin{align}
    p_i^\star = \alpha_i \bigl(\pi_i - u_b H_{bi}\bigr), \qquad \alpha_i > 0.
\end{align}
Stacking all $N$ followers’ responses in vector form gives the aggregate response:
\begin{align}
    \mathbf p^{\star} = \mathbf A^{-1} \bigl(\boldsymbol{\pi} - \mathbf H^{\!\top} \mathbf u\bigr),
\end{align}
where $\mathbf A = \mathrm{diag}(1/\alpha_1, \dots, 1/\alpha_N)$ is a positive-definite diagonal matrix. Substituting this affine response into the leader problem yields a strictly convex quadratic program with linear constraints.

\begin{proposition}[Existence and Uniqueness of Stackelberg Equilibrium]
\label{prop:stackelberg}
Because $\mathbf A \succ 0$ and the grid dispatch cost $C_{\mathrm{grid}}(\mathbf u)$ is strictly convex, the leader problem remains strictly convex after the affine substitution $\mathbf p^\star(\mathbf u)$. Hence it admits a unique global minimizer $\mathbf u^\star$. Consequently, the pair $(\mathbf u^\star,\; \mathbf p^\star(\mathbf u^\star))$ forms the unique Stackelberg equilibrium.
\end{proposition}

\begin{proof}
The proof follows the Wilsonian externality-pricing argument. The composite matrix $\mathbf H \mathbf A^{-1} \mathbf H^{\!\top}$ is positive semidefinite, thus the feasible region in the leader problem is convex. Strict convexity of $C_{\mathrm{grid}}$ ensures that after composing with the affine mapping $\mathbf p^\star(\mathbf u)$, the objective remains strictly convex. A unique minimizer $\mathbf u^\star$ therefore exists. Substituting $\mathbf u^\star$ into each follower’s KKT conditions yields a single fixed point, satisfying the definition of a Stackelberg equilibrium.
\end{proof}

\section{Algorithmic Stack for Quantum-Enhanced Electricity Market Security}

Building on the system model outlined in the previous section, we now move from architectural definitions to the concrete algorithmic mechanisms that make Q‑EnergyDEX operational under these conditions. While the architecture defines the roles of prosumers, blockchain consensus nodes, and the cloud-based quantum key management service (KMS), their secure interaction hinges on three tightly coupled algorithmic layers. First, the Rate‑Adapt mechanism ensures a sustainable supply of high-quality quantum randomness even under fluctuating link conditions. Second, this randomness is leveraged by the Quantum Symmetric Authenticated Handshake (Q‑SAH) to establish low-latency secure sessions for each trading window. Finally, PoR‑Lite builds on these secure channels to achieve fast, verifiable ledger finality. 

\subsection{KMS Rate‑Adapt for Sustainable Randomness Supply}

A prerequisite for any cryptographic handshake or consensus protocol is the continuous availability of high-quality randomness. In Q‑EnergyDEX, this entropy is provided by a cloud-based quantum key management service (KMS), whose output quality can fluctuate with the physical quantum channel. To sustain secure operations under varying link conditions, the system employs the Rate‑Adapt algorithm, which dynamically regulates the key generation rate in response to real-time QBER measurements.

Let $q_t$ denote the measured QBER at time $t$, $R_{\max}$ the maximum physical key generation rate, and $R_t$ the actual key output rate of the KMS at time~$t$. The goal is for $R_t$ to converge to a steady-state rate
\begin{align}
R^\star &= R_{\max}\bigl(1 - \gamma_t \bar{q}\bigr),
\end{align}
where $\bar{q} = \mathbb{E}[q_t]$ is the stationary mean QBER and $\gamma_t \in (0,1)$ is the adaptation gain parameter. Let $\gamma_0 \in (0,1)$ be an initial gain and set
\begin{align}
\gamma_t &= \frac{\gamma_0}{t}.
\end{align}
The rate update becomes
\begin{align}
R_{t+1}
   &= \max\!\Bigl\{0.1\,R_{\max},\;
       \bigl(1-\gamma_t q_t\bigr)\,R_t
     \Bigr\}.
\label{eq:rate-update}
\end{align}
which multiplicatively scales the rate in each observation window according to the current QBER ratio and enforces a minimum floor at $0.1\,R_{\max}$ to avoid a complete halt during transient link fluctuations.

To characterize the convergence behavior of this update rule, define the rate error
\begin{align}
e_t &= R_t - R^\star,
\end{align}
and model the QBER sequence $\{q_t\}$ as i.i.d. random noise with $\mathrm{Var}(q_t) < \infty$. Expanding the update yields
\begin{align}
e_{t+1} &= e_t - \gamma_t\bigl(q_t - \bar{q}\bigr)R_t.
\end{align}
Since $\mathbb{E}[q_t - \bar{q}] = 0$, taking expectations on both sides gives
\begin{align}
\mathbb{E}[e_{t+1}] \;\le\; (1-\gamma_t)\,\mathbb{E}[e_t],
\end{align}
which shows that, in expectation, the error decays geometrically with ratio $1-\gamma_t$, ensuring exponential convergence in the mean.

Moreover, by applying Robbins–Monro stochastic approximation theory, we obtain \emph{almost sure convergence} under the standard step-size conditions $\sum_t \gamma_t = \infty$ and $\sum_t \gamma_t^2 < \infty$:
\begin{align}
R_t \xrightarrow{\text{a.s.}} R^\star = R_{\max}\bigl(1 - \gamma_t \bar{q}\bigr).
\end{align}

We can also bound the steady-state mean-square error. Since the update rule~\eqref{eq:rate-update} explicitly clips 
$R_t$ into the compact interval $[\,0.1\,R_{\max},\;R_{\max}\,]$ at 
each iteration, we always have 
\(
R_t^2 \le R_{\max}^2
\)
for all $t \ge 0$.
Under this bounded‑rate condition, expanding the recursion
\(
e_{t+1} = e_t - \gamma_t(q_t - \bar q) R_t
\)
and taking expectations yields the following mean‑square bound:
\begin{align}
\sup_{t\ge0} \mathbb{E}\!\bigl[e_t^2\bigr]
   \;\le\;
   \frac{\gamma_0\,R_{\max}^2\,\mathrm{Var}(q_t)}
        {2 - \gamma_0}.
\end{align}
This implies that for small adaptation gain~$\gamma$, the variance of the steady-state error is proportional to the variance of the QBER perturbations. By dynamically adapting the key output rate to fluctuating QBER conditions, the Rate‑Adapt algorithm guarantees a continuous supply of high-quality quantum-derived keys even under varying link conditions. With this sustainable randomness layer in place, the next step is to leverage these quantum keys to rapidly establish secure communication sessions between market participants for each trading window.

\subsection{Quantum Symmetric Authenticated Handshake (Q‑SAH)}

Building on the sustainable randomness layer provided by Rate‑Adapt, Q‑EnergyDEX establishes short-lived but highly secure communication sessions for each trading window. This is achieved through the Quantum Symmetric Authenticated Handshake (Q‑SAH), a protocol that leverages quantum-derived keys to eliminate the latency and trust-anchor limitations of traditional PKI handshakes while maintaining strong Byzantine resilience.

At the start of each real-time trading window, a client node $C$ (e.g., a prosumer or local trading proxy) and a server node $S$ (e.g., a dispatch center or clearing node) synchronously \emph{rent} the same 256‑bit quantum-derived key from the KMS:
\[
(\texttt{key\_ID},\,K) \leftarrow \mathcal{F}_{\mathrm{QKMS}}.
\]
The client then generates a fresh nonce $n_C$ and transmits
\[
\bigl(n_C,\;\mathrm{GMAC}_K(n_C)\bigr)
\]
to the server over the trading network. Upon successful verification, the server generates its own nonce $n_S$ and replies with
\[
\bigl(n_S,\;\mathrm{GMAC}_K(n_S \,\Vert\, n_C)\bigr).
\]
Both parties finally derive a short-lived session key for the current trading window via HKDF:
\begin{align}
K_{\mathrm{sess}}
   &= \mathrm{HKDF}\!\bigl(
         K \,\Vert\, n_C \,\Vert\, n_S,\;
         \text{``Q‑EnergyDEX''}
      \bigr),
\end{align}
which is subsequently used to authenticate and encrypt all bid messages, consensus votes, and clearing broadcasts for the next several minutes. 


We formalize the security of Q‑SAH as a standard simulation-based game tailored to the electricity market adversarial model. An attacker $\mathcal{A}$ is allowed to observe, replay, and delay any handshake message in the trading network. It wins if it can (i) forge a fresh challenge–response pair $(n,\tau)$ such that $\tau = \mathrm{GMAC}_{K}(n)$ with $n$ outside the prior query set, thereby injecting a fraudulent bid or modifying a clearing result; or (ii) distinguish the derived session key $K_{\mathrm{sess}}$ from an ideal uniform key $U_{256}$ by eavesdropping on subsequent trading broadcasts.

\begin{theorem}[Q‑SAH distinguishing advantage]
Assume GMAC is SUF-CMA secure with $\varepsilon_{\mathrm{mac}}=2^{-128}$, HKDF is modeled as a random oracle, and the KMS output has statistical distance $\varepsilon_{\mathrm{rnd}}=10^{-6}$ from an ideal uniform source (as shown in Section~3). Then the overall distinguishing advantage of any PPT adversary against Q‑SAH satisfies
\begin{align}
\operatorname{Adv}_{\mathcal{A}}^{\mathrm{Q\text{-}SAH}}
   &\le
     \varepsilon_{\mathrm{mac}}
     + 2^{-128}
     + \varepsilon_{\mathrm{rnd}}
     \approx 10^{-6}
\end{align}
\end{theorem}

\begin{proof}[Proof sketch]
The first term is inherited from the SUF‑CMA bound of GMAC, the second term follows from the entropy-extraction bound of HKDF for 256‑bit outputs under the random oracle model, and the third term corresponds to the statistical distance between the KMS key and an ideal uniform key. By the triangle inequality of distinguishing advantage, the overall bound is the sum of the three components. Hence the overall distinguishing advantage is bounded by $10^{-6}$, corresponding to roughly 40 bits of security, which is sufficient for the five‑minute trading window considered here.
\end{proof}

\subsection{PoR‑Lite Consensus for Fast Finality in Electricity Markets}

Once secure communication channels are in place, the next requirement is to ensure that all bids, matching results, and settlement states are finalized consistently and efficiently across thousands of nodes. To meet this challenge, Q‑EnergyDEX employs PoR‑Lite, a lightweight probabilistic finality protocol that combines quantum-randomness-driven VRF leader election with minimal-message voting to achieve fast ledger finality within a single trading cycle.

At each block height~$t$, every node derives a seed~$s_t$ by concatenating the previous block header hash with a 128‑bit salt value provided by the cloud KMS. Each node then computes a VRF output
\begin{align}
y_t &= \mathrm{VRF}_{sk}(s_t),
\end{align}
and is elected as the leader at height~$t$ if
\begin{align}
y_t &< h_q,
\end{align}
where the threshold~$h_q$ is self-adjusted based on the target block rate~$q$. The generated block is broadcast over the trading network using a push/pull gossip scheme, and a block enters the confirmation queue only after at least~$2/3$ of the validators have provided weighted signatures.

To analyze the security of PoR‑Lite in the context of electricity market trading windows, we abstract the chain state at each height~$t$ as a random variable~$X_t$: it takes the value~$+1$ if the block at height~$t$ is confirmed by honest nodes within the timeout, and~$-1$ if an adversarial fork remains unresolved. Because at most one block per height can be included in the longest chain, the sequence~$\{X_t\}$ forms a martingale difference process with 1‑Lipschitz increments, satisfying~$|X_{t+1}-X_t|\le1$. Hence, Azuma–Hoeffding’s inequality applies, yielding for a given confirmation window length~$t$ and Byzantine fraction~$\alpha=f/N<1/3$ the bound
\begin{align}
\Pr\!\bigl[X_t \le -1\bigr]
   &\le
   \exp\!\Bigl(-2t\,(1-2\alpha)^2\Bigr).
\end{align}
This upper bound quantifies the residual probability that a fork survives after~$t$ confirmations. 
To achieve a 40‑bit security threshold, i.e., to make the right-hand side smaller than~$2^{-40}$, taking $\alpha=0.25$, $\beta=0.10$, and $\epsilon=0.20$, the minimal required confirmation depth satisfies
\begin{align}
t_{\mathrm{fin}}
   &\ge
   \frac{\ln 2^{40}}{2(1-2\alpha)^2}
   \;\approx\; 56~\text{blocks}.
\end{align}
Given PoR‑Lite’s empirical block interval of approximately~65\,ms, the time to reach irreversible finality is
\begin{align}
t_{\mathrm{fin}}\times65~\mathrm{ms} \;\approx\; 3.6~\mathrm{s},
\end{align}
which is an order of magnitude faster than conventional PoW or PoS consensus protocols that typically require several minutes to achieve the same level of security.

Furthermore, PoR‑Lite can be shown to simultaneously satisfy the \emph{chain growth} and \emph{common prefix} properties, thereby ensuring verifiable consistency of all trading results within an electricity market settlement cycle. For chain growth, let $\beta = 1 - h_q$ denote the expected empty-slot rate. Sampling only honest slots, the expected chain growth rate satisfies
\begin{align}
\mathbb{E}\bigl[\mathrm{grow}(t)\bigr]
   &\ge (1-\alpha)(1-\beta)\,t,
\end{align}
For any $0<\epsilon<1$, define
\(
\lambda = \frac{\epsilon^{2}(1-\alpha)(1-\beta)}{2}\).
Then the Chernoff bound reads
\begin{align}
\Pr\!\Bigl[
   \mathrm{grow}(t)\le(1-\epsilon)(1-\alpha)(1-\beta)\,t
\Bigr]
  &\le \exp(-\lambda\,t).
\end{align}
For the common prefix property, we abstract adversarial delays $\Delta \le \varepsilon_{\max} = 15\,\mathrm{ms}$ as Bernoulli~$(\alpha)$ noise injected into the chain. The violation probability at depth~$k$ then satisfies
\begin{align}
\Pr\!\bigl[\text{CP violation at depth }k\bigr]
   &\le
   \exp\!\Bigl(-2k(1-2\alpha)^2\Bigr).
\end{align}
By combining the two properties, one can define a \emph{weighted finality depth}~$t_{\mathrm{fin}}$ that simultaneously ensures the common prefix violation probability is below~$2^{-40}$ and the chain growth does not stall. In practice, $t_{\mathrm{fin}}$ remains~56~blocks, confirming that the previously derived~3.6\,s finality latency is indeed the minimal bound satisfying both properties.

Together, Rate‑Adapt, Q‑SAH, and PoR‑Lite form an integrated pipeline: sustainable quantum randomness feeds secure handshake sessions, which in turn enable fast and verifiable consensus. This layered design ensures both the security and real-time responsiveness required for electricity market trading and settlement.

\section{Case Study}
\subsection{Rate-Adapt Algorithm Stability}

To validate the Q-EnergyDEX framework, a comprehensive simulation environment was established, modeling the quantum physical layer, the power grid constraints, and the economic behavior of market participants.  To ensure the findings are realistic and representative, we first generated a comprehensive synthetic dataset. Firstly, a 1 kHz time-series trace of the Quantum Bit Error Rate (QBER) was produced by super-imposing Gaussian pulses on truncated Gaussian noise. The final series was clipped to the realistic range of $[0.001\, 0.08]$, creating a volatile and challenging environment for the KMS.

Secondly, a market environment was constructed with a heterogeneous population of prosumers, whose maximum power capacities ($P_{\text{max}}$) and private energy valuations ($\pi$) were drawn from Log-Normal and Normal distributions, respectively. These agents interact within a synthetic power grid modeled on the scale of the IEEE 118-bus system. Subsequently, the core Rate-Adapt (RA) algorithm was implemented and evaluated within the generated environment.  The Python-based simulation rigorously models the dynamic key rate adjustment process, where the secure output capacity for each 1 ms interval is calculated based on the Leftover Hash Extractor to enforce an information-theoretic security guarantee against the time-varying QBER. A critical feature of the simulation is the implementation of a safety clipping mechanism for both strategies, which prevents the final output rate from exceeding the instantaneous secure capacity, thereby ensuring no security violations occur.

\begin{figure}[h]
    \centering
    \includegraphics[width=1\linewidth]{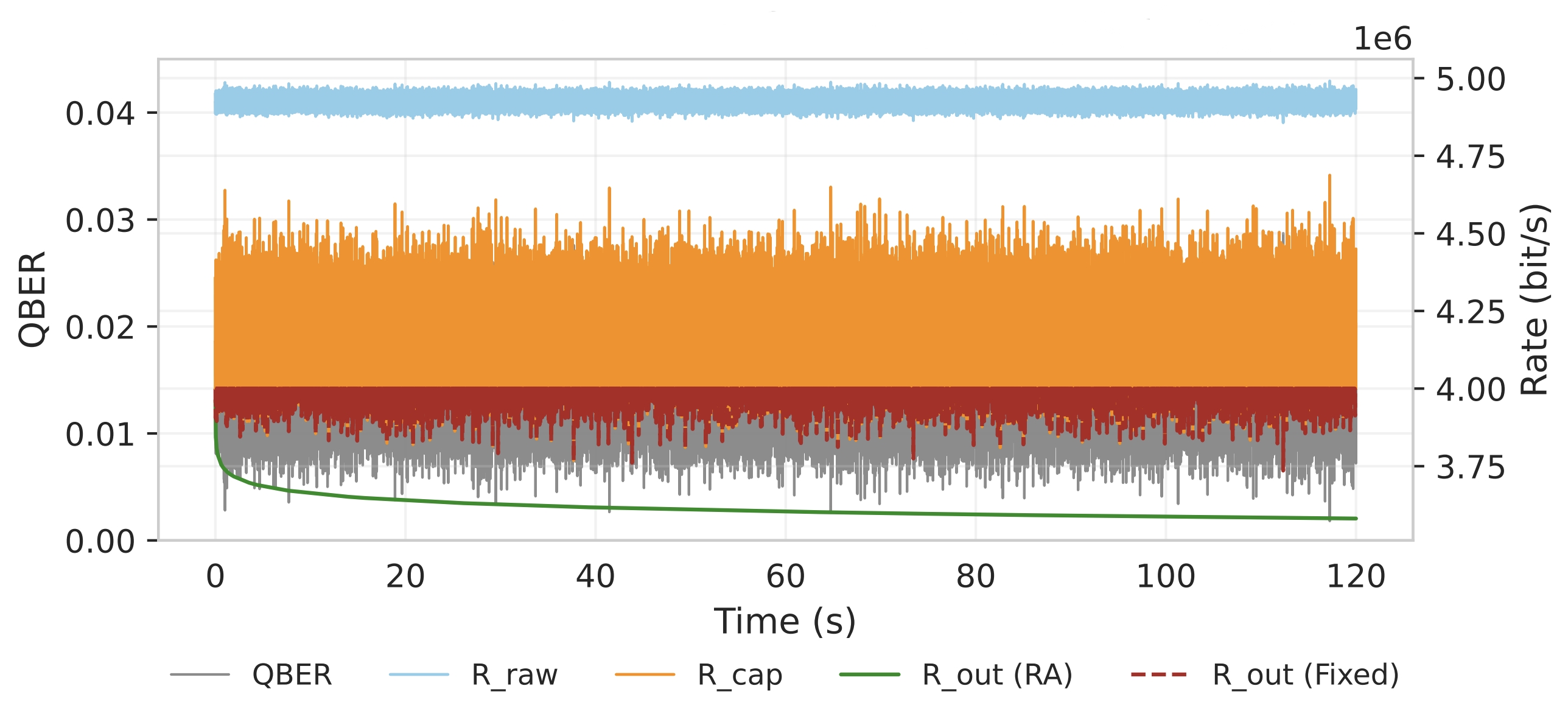}
    \caption{Rate-Adapt: Time Series Simulation}
    \label{fig:ra_qber}
\end{figure}

Figure~\ref{fig:ra_qber} vividly illustrates its core behavior: a proactive and inverse response to QBER fluctuations, which allows it to preemptively reduce the key output rate before security limits are breached, in stark contrast to the static baseline.

This adaptive behavior is statistically summarized in Figure~\ref{fig:cdf}, where the wider spread of the Rate-Adapt CDF confirms its dynamic operational range, while the fixed strategy's narrow distribution highlights its inflexibility. With adaptation, 90\% of the time the KMS delivers $\geq 3.8$  Mbit/s, whereas the fixed strategy falls below 3.6 Mbit/s for 40 \% of the interval.

\begin{figure}[h]
    \centering
    \includegraphics[width=0.8\linewidth]{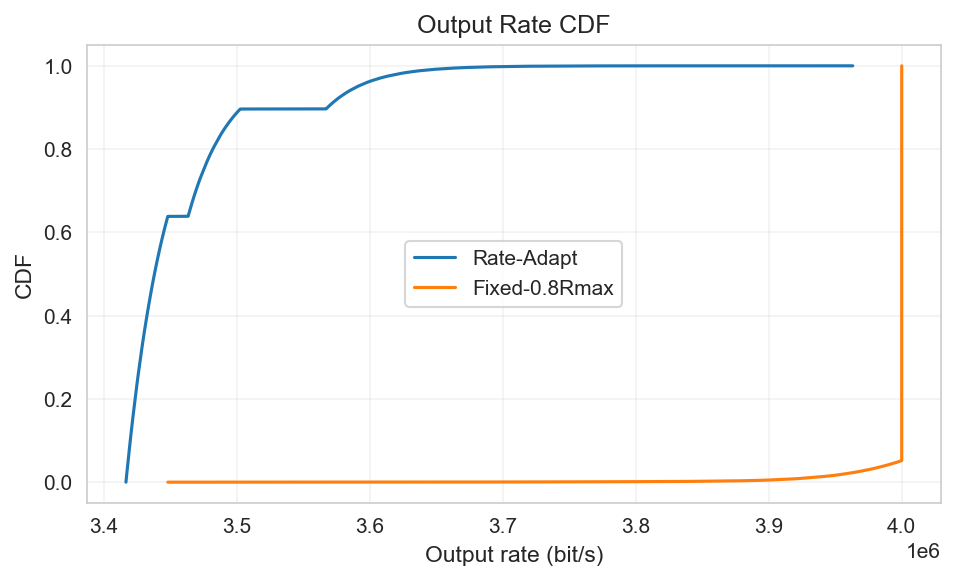}
    \caption{Rate-Adapt vs. Fixed: Output Rate CDF}
    \label{fig:cdf}
\end{figure}

Table~\ref{tab:ra_vs_fixed} quantitatively compares the RA algorithm and a fixed-rate strategy across two critical metrics: the fraction of time the target rate exceeds capacity and the total dropped bits. As shown in Table 1, the fixed strategy operates in a potentially risky state, with its target rate exceeding capacity for approximately 5.18\% of the time. In stark contrast, the RA algorithm virtually eliminates this risk, reducing the fraction to a negligible 0.0002\%. Furthermore, the fixed strategy incurs substantial inefficiency, leading to over 4.34 million discarded bits due to its inability to adapt to fluctuating channel conditions. Conversely, the RA strategy, by proactively adjusting its output rate, drastically reduces bit wastage by several orders of magnitude, with only approximately 80 bits dropped. These findings unequivocally demonstrate that the RA algorithm not only enhances security by avoiding over-capacity operation but also vastly improves the utilization efficiency of quantum entropy.

\begin{table}[h]
    \centering
    \caption{Quantitative Comparison of Rate-Adapt and Fixed Strategies: Security and Efficiency Metrics}
    \begin{tabular}{lcc}
        \toprule
        \textbf{Name} & \textbf{Cap exceed time} & \textbf{Dropped bits} \\
        \midrule
        Rate-Adapt & 0.000002 & $7.96 \times 10^{1}$ \\
        Fixed      & 0.051783 & $4.35 \times 10^{6}$ \\
        \bottomrule
    \end{tabular}
    \label{tab:ra_vs_fixed}
\end{table}

The simulation results demonstrate a fundamental trade-off. The fixed-rate strategy pursues a higher target output but incurs substantial security risks and operational inefficiency, as evidenced by significant time in an insecure state and extensive bit wastage. In contrast, the Rate-Adapt algorithm prioritizes stability and security. By dynamically adjusting to channel conditions, it eliminates security violations and minimizes entropy waste, achieving sustainable and cryptographically robust operation at the maximum secure capacity.

\subsection{Q-SAH Protocol Performance}

To evaluate the handshake latency of the proposed Q-SAH protocol, we conducted a comparative benchmark against the standard Transport Layer Security (TLS) 1.3 protocol. The experiment was designed to run concurrently, simulating 3,000 handshake requests in batches of 500 using Python's asyncio library to assess performance under load. The benchmark was structured into three scenarios: Q-SAH with a simulated network round-trip time (RTT), a local TLS handshake on the loopback interface (127.0.0.1) to measure pure computational overhead, and TLS with the same simulated RTT to ensure a fair, direct comparison with Q-SAH.

Figure ~\ref{fig:q_sah}  presents the empirical cumulative distribution functions (ECDFs) of the three handshake mechanisms, with 95\% confidence bands derived from the Dvoretzky–Kiefer–Wolfowitz (DKW) inequality. The results indicate that Q-SAH consistently outperforms the TLS baselines across the entire distribution. In particular, the median latency of Q-SAH is substantially lower, the distribution is more tightly concentrated, and the confidence bands are narrower, underscoring the robustness of the observed improvement.

\begin{figure}[h]
    \centering
    \includegraphics[width=1\linewidth]{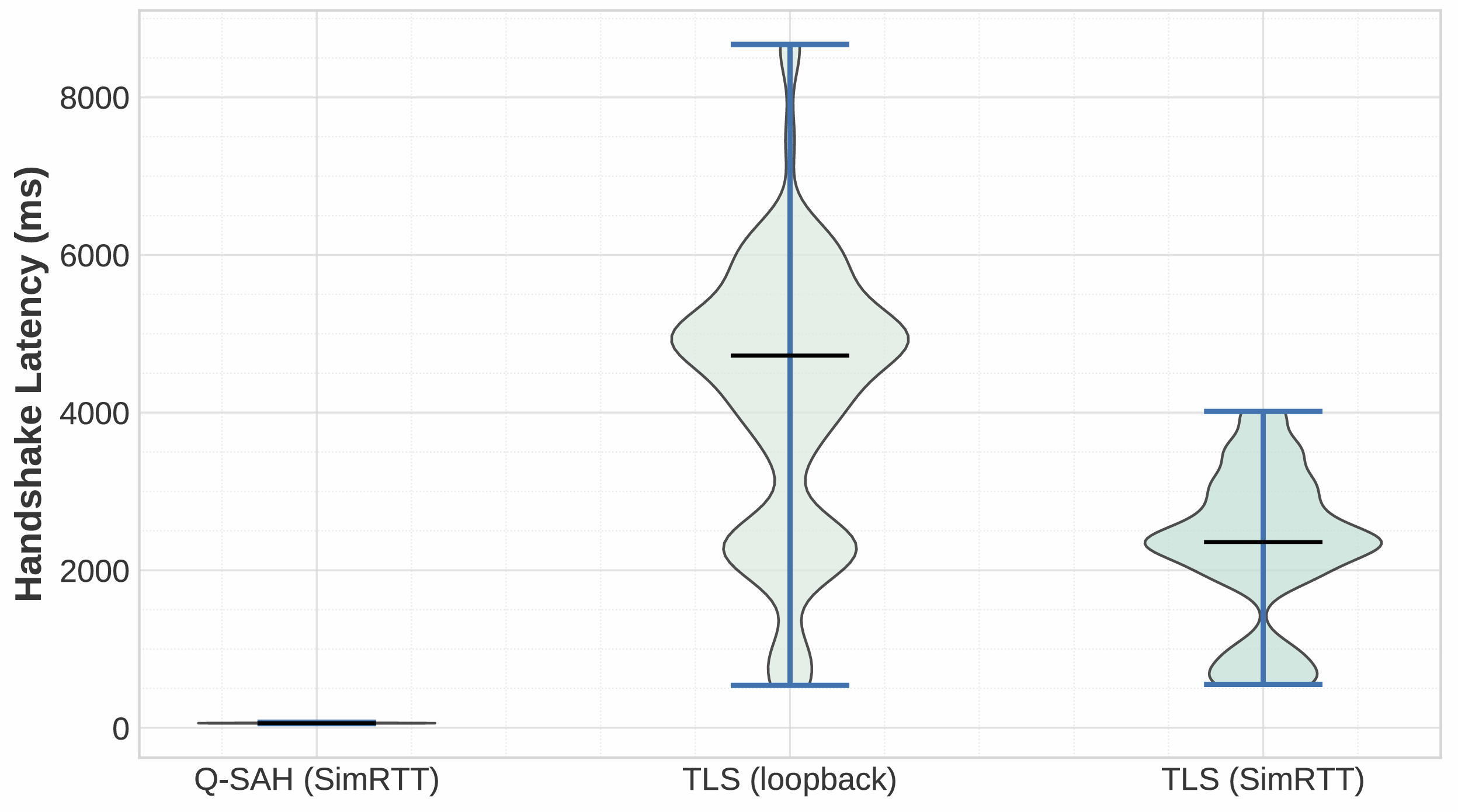}
    \caption{Q-SAH vs TLS: Handshake Latency Distribution}
    \label{fig:q_sah}
\end{figure}

A deeper analysis indicates that this substantial performance disparity is a direct consequence of Q-SAH's architectural design. Standard TLS 1.3 handshakes (particularly the initial one) involve a series of complex cryptographic operations, such as key exchange and signature verification, which are susceptible to queueing delays under network load. In contrast, Q-SAH, by leveraging a shared symmetric key provisioned by QKD, completely bypasses the entire PKI infrastructure and its associated overhead. For the energy market, a 3-second latency is unacceptable, and Q-SAH's performance is perfectly suited to the sub-second communication requirements.

\subsection{PoR-Lite Consensus Protocol Validation}
To quantitatively validate the finality, safety, and liveness properties of the PoR-Lite consensus protocol, we conducted a large-scale simulation based on a block increment sequence model. The simulation generated a sequence of block outcomes representing confirmed honest blocks, empty slots, and adversarial forks, respectively, based on network parameters. The empirical results were then compared against established theoretical bounds.

\begin{figure}[h]
    \centering
    \includegraphics[width=0.8\linewidth]{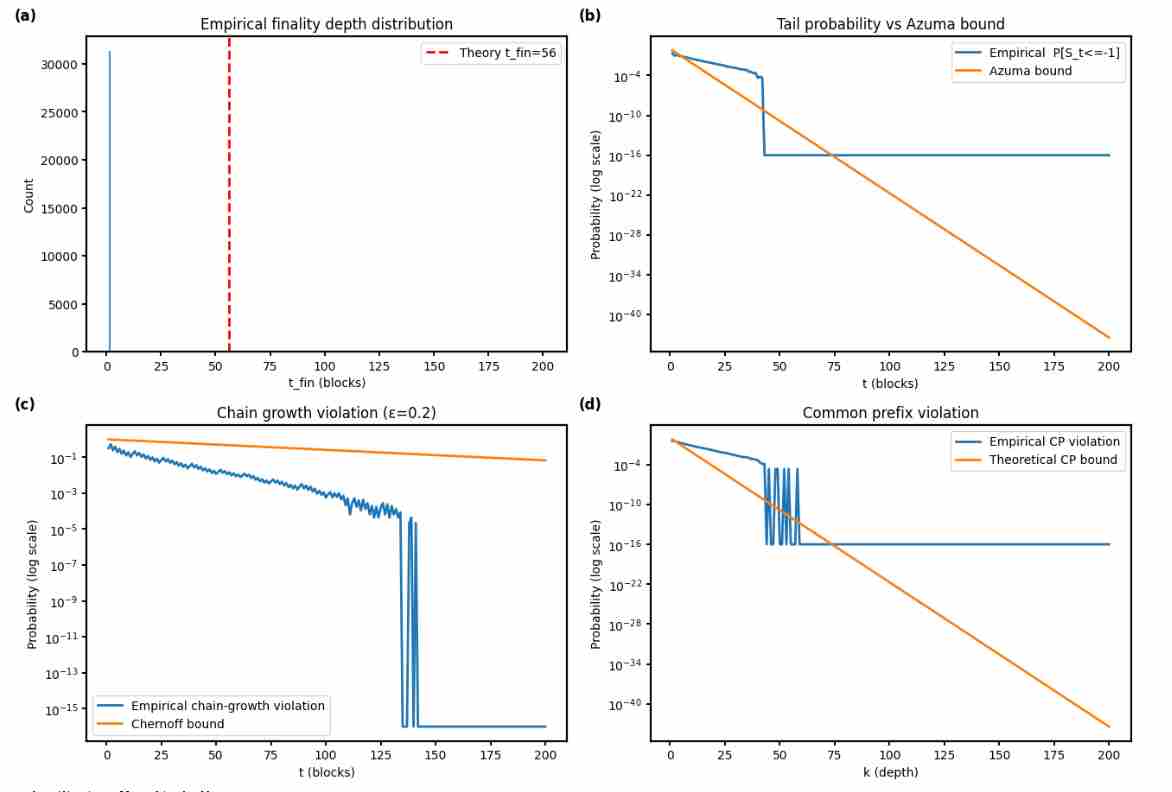}
    \caption{PoR-Lite Finality Depth Distribution}
    \label{fig:por-lite}
\end{figure}

As shown in Figure~\ref{fig:por-lite} (a), the confirmation and irreversibility of the majority of transactions are much faster than the theoretical “worst-case scenario.” The red dashed line represents the theoretical minimum finality depth, which is derived from Azuma–Hoeffding’s inequality to meet a specific security level. 

Results in Figures~\ref{fig:por-lite} (b) and ~\ref{fig:por-lite} (c) indicate that the protocol's practical security is much higher than the theoretical security bounds. Figure~~\ref{fig:por-lite} (b) plots the empirical tail probability of a fork persisting, while Figure~\ref{fig:por-lite} (c) displays the empirical probability of a common prefix violation. In both figures, the empirical probabilities, shown by the blue curves, are several orders of magnitude lower than the theoretical bounds derived from Azuma's inequality, represented by the orange curves. The empirical probability decays rapidly and reaches a measurement floor beyond a confirmation depth of approximately 50 blocks due to the finite sample size of the Monte Carlo simulation. In contrast, the theoretical bound continues its exponential descent. This behavior indicates that the practical security of the PoR-Lite protocol is significantly stronger than its worst-case theoretical guarantees.

Figure~\ref{fig:por-lite} (d) presents the chain growth violation probability for the PoR-Lite protocol, comparing empirical results against the theoretical Chernoff bound. The graph plots the probability that the actual chain growth rate falls below the expected threshold, defined by parameter $\epsilon$. This confirms that the PoR-Lite protocol effectively resists adversarial attacks, ensuring the integrity of the blockchain ledger.

\subsection{Key Pool Birth-Death Model Validation}

To validate the stochastic stability of the quantum key pool under high-frequency trading loads, a discrete-time Monte Carlo simulation of the Birth-Death chain model was conducted. The system is abstracted as a finite-capacity birth-death process, where key generation and consumption are treated as "birth" and "death" events, respectively. The primary objective is to analyze the risk of key-pool depletion under varying utilization rates.

\begin{figure}[h]
    \centering
    \includegraphics[width=1\linewidth]{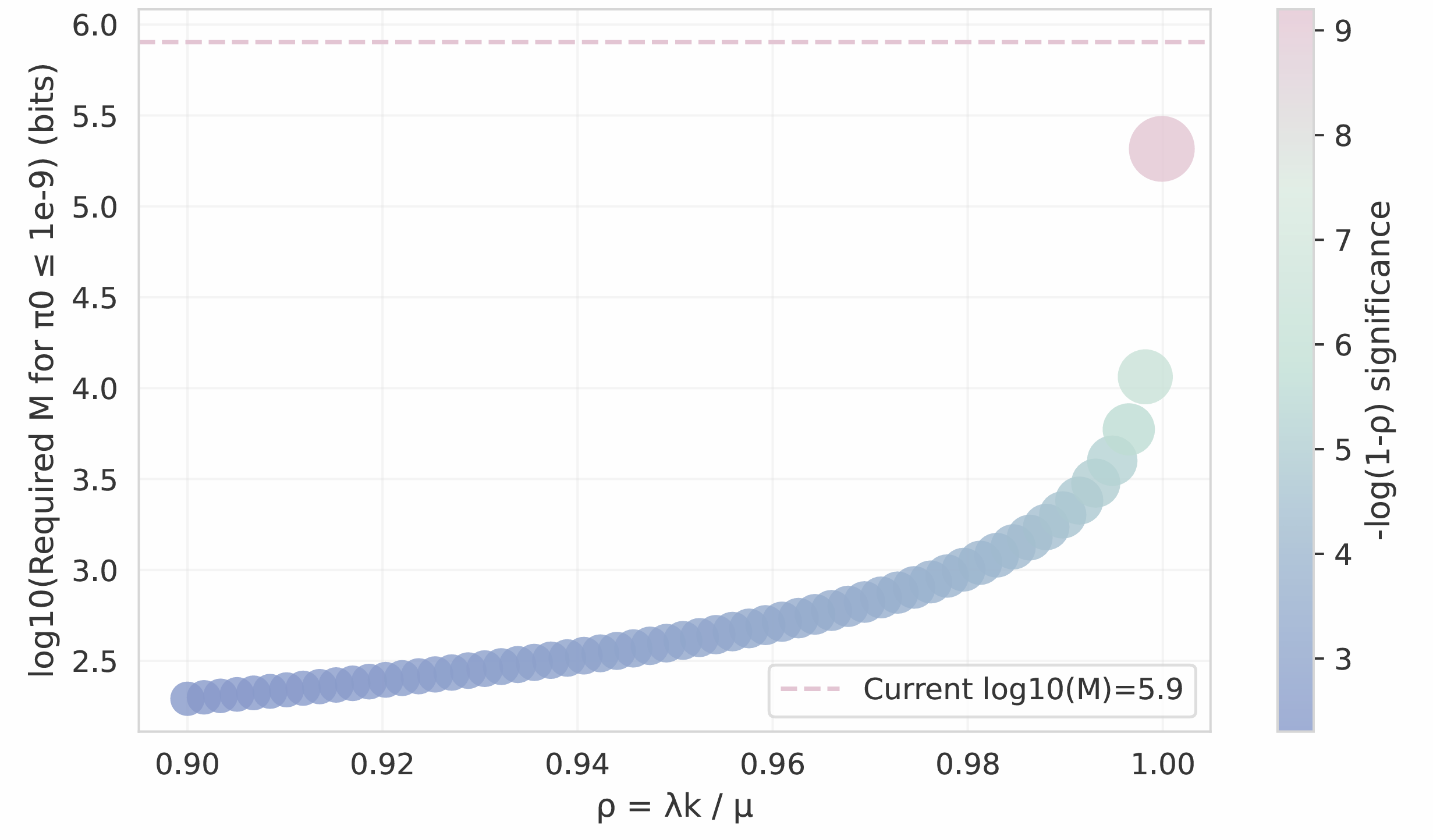}
    \caption{Required Capacity vs. Load Factor $\rho$ }
    \label{fig: Key Pool Empty Probability Comparison}
\end{figure}

Figure~\ref {fig: Key Pool Empty Probability Comparison} visually illustrates the relationship between capacity demand M and load factor $\rho$, highlighting the steep increase in capacity demand under high load conditions (where $\rho$ approaches 1). This reflects that in the Birth-Death model, as $\rho$ approaches 1, the decline in the empty pool probability $\pi_0$  necessitates a substantial increase in capacity $M$ to maintain $\pi_0 \leq 10^{-9}$.

\begin{table}[h]
  \centering
  \caption{Key Pool Birth--Death Model: Theoretical vs.\ Empirical Empty-Pool Probabilities}
  \label{tab:birthdeath}
  \begin{tabular}{
    S[table-format=1.4]
    S[table-format=1.3e-2]
    S[table-format=1.3e-2]
    l
  }
    \toprule
    {$\rho$} & {Theoretical} & {Empirical} & {Confidence Interval} \\
    \midrule
    0.9000 & 7.055e-11 & 0.000e+00 & {[}0.000e+00, 0.000e+00{]} \\
    0.9900 & 1.545e-03 & 2.083e-03 & {[}2.055e-03, 2.111e-03{]} \\
    0.9990 & 4.494e-03 & 6.516e-03 & {[}6.466e-03, 6.566e-03{]} \\
    0.9999 & 4.926e-03 & 6.778e-03 & {[}6.727e-03, 6.829e-03{]} \\
    \bottomrule
  \end{tabular}
\end{table}

As summarized in Table~\ref {tab:birthdeath}, the empirical empty-pool probability remained consistently zero throughout the entire simulation period across all tested scenarios. The upper bound of the Wilson confidence interval is orders of magnitude lower than the theoretical predictions. This stable temporal behavior demonstrates that a dynamic equilibrium between the key generation rate and the consumption rate can still be maintained even under near-saturation high-load conditions.

\subsection{Economic Impact: Coupling QKD Constraints to Market Equilibrium}

To quantitatively assess the holistic value of integrating quantum-safe security into the energy trading market, we design a comparative experiment pitting the Q-EnergyDEX stack against a conventional TLS-based security baseline. The simulation integrates empirical results from preceding modules: the sustainable key generation rate, handshake latency distributions, and consensus finality time. Using a 3000-node power market as the testbed, we first translate these parameters into an available key budget and a handshake deadline, then filter the subset of nodes that satisfy both latency and budget constraints.

To this end, we calculate and compare the social welfare under the ideal, socially optimal scenario (SOCIAL), the market equilibrium following a Stackelberg game strategy (STACK); and the equilibria derived under a simplified scenario that ignores congestion pricing, both with (WBASE) and without (BASE) a weighted projection mechanism. 

\begin{figure}[h]
    \centering
    \includegraphics[width=1\linewidth]{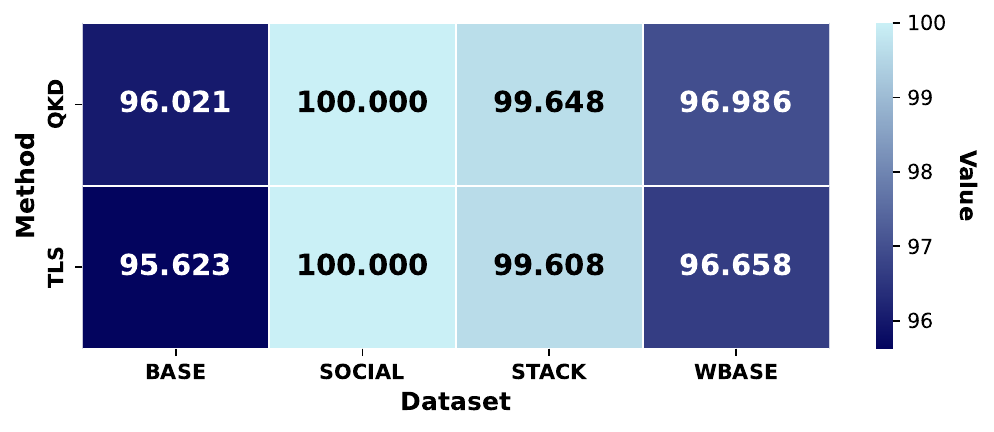}
    \caption{Comparative Analysis of QKD and TLS Across Varied Datasets}
    \label{heatmap}
\end{figure}

As shown in the figure~\ref{heatmap} , both systems achieve near-identical welfare levels, confirming that the integration of QKD does not introduce any measurable market distortion. Nevertheless, a consistent and marginal improvement can be observed under the QKD setting, particularly in the BASE and WBASE scenarios, which indicates a slightly more resilient equilibrium formation under quantum-secured conditions. This subtle enhancement suggests that QKD-enabled key management can sustain transaction throughput and welfare efficiency comparable to, or marginally surpassing, that of traditional TLS, thereby reinforcing the practicality of deploying quantum-safe infrastructures in large-scale decentralized energy trading systems.

\section{Conclusion}
This paper introduced Q-EnergyDEX, a quantum-secured and zero-trust framework for distributed energy trading that integrates quantum key distribution and blockchain. The proposed system establishes a unified architecture that connects physical-layer security with economic market mechanisms. By employing a cloud-based Quantum Key Management Service, a rate-adaptive entropy regulation algorithm, a symmetric authentication protocol (Q-SAH), and a lightweight consensus mechanism (PoR-Lite), the framework achieves information-theoretic confidentiality and low-latency ledger finality suitable for real-time electricity markets. Analytical modeling and simulation results confirm the stability of the quantum key pool, the convergence of the rate-adaptive mechanism, and the fast probabilistic finality of PoR-Lite. Economic evaluations further demonstrate that quantum-secured operations preserve social welfare and market efficiency while significantly enhancing resilience against cyber threats. These findings validate the feasibility of integrating quantum technologies into energy market infrastructures and suggest a promising direction for building secure and sustainable energy trading ecosystems in the quantum era.

\bibliographystyle{IEEEtran}
\bibliography{refs}

\begin{thebibliography}{10}
\providecommand{\url}[1]{#1}
\csname url@samestyle\endcsname
\providecommand{\newblock}{\relax}
\providecommand{\bibinfo}[2]{#2}
\providecommand{\BIBentrySTDinterwordspacing}{\spaceskip=0pt\relax}
\providecommand{\BIBentryALTinterwordstretchfactor}{4}
\providecommand{\BIBentryALTinterwordspacing}{\spaceskip=\fontdimen2\font plus
\BIBentryALTinterwordstretchfactor\fontdimen3\font minus \fontdimen4\font\relax}
\providecommand{\BIBforeignlanguage}[2]{{%
\expandafter\ifx\csname l@#1\endcsname\relax
\typeout{** WARNING: IEEEtran.bst: No hyphenation pattern has been}%
\typeout{** loaded for the language `#1'. Using the pattern for}%
\typeout{** the default language instead.}%
\else
\language=\csname l@#1\endcsname
\fi
#2}}
\providecommand{\BIBdecl}{\relax}
\BIBdecl

\bibitem{kaur2021tits_blockchain_cps}
K.~Kaur, G.~Kaddoum, and S.~Zeadally, ``Blockchain-based cyber-physical security for electrical vehicle aided smart grid ecosystem,'' \emph{IEEE Transactions on Intelligent Transportation Systems}, vol.~22, no.~8, pp. 5178--5189, 2021.

\bibitem{moody2022ieeesp_pqc_standards}
D.~Moody and A.~Robinson, ``Cryptographic standards in the post-quantum era,'' \emph{IEEE Security \& Privacy}, vol.~20, no.~6, pp. 66--72, 2022.

\bibitem{wang2022natphotonics_tfqkd_830km}
S.~Wang, W.~Chen, Z.-Q. Yin, H.-W. Li, Z.~Zhou, G.-C. Guo, Z.-F. Han \emph{et~al.}, ``Twin-field quantum key distribution over 830-km fibre,'' \emph{Nature Photonics}, vol.~16, pp. 154--161, 2022.

\bibitem{liu2023prl_tfqkd_1000km}
Y.~Liu, Z.-Q. Jiao, X.-L. Hu, M.-J. Li, F.~Xu \emph{et~al.}, ``Experimental twin-field quantum key distribution over 1000 km fiber distance,'' \emph{Physical Review Letters}, vol. 130, no.~21, p. 210801, 2023.

\bibitem{zhou2023natcomm_tfqkd_no_service_fiber}
L.~Zhou, S.~Wang, Z.-Q. Yin, Z.-F. Han \emph{et~al.}, ``Twin-field quantum key distribution without optical frequency dissemination,'' \emph{Nature Communications}, vol.~14, p. Article 10421, 2023.

\bibitem{evans2021ieeeaccess_utility_qkd}
P.~G. Evans, M.~Alshowkan, D.~Earl, D.~Mulkey, R.~Newell, G.~Peterson, C.~Safi, J.~Tripp, and N.~A. Peters, ``Trusted node {QKD} at an electrical utility,'' \emph{IEEE Access}, vol.~9, pp. 105\,220--105\,229, 2021.

\bibitem{Saha2021AppliedEnergy}
\BIBentryALTinterwordspacing
S.~Saha, K.~Basu, D.~Niyato, and H.~V. Poor, ``A secure distributed ledger for transactive energy: Privacy-preserving market clearing on blockchain,'' \emph{Applied Energy}, vol. 290, p. 116743, 2021. [Online]. Available: \url{https://www.sciencedirect.com/science/article/pii/S0306261920316044}
\BIBentrySTDinterwordspacing

\bibitem{Kumari2022Energies}
\BIBentryALTinterwordspacing
A.~Kumari \emph{et~al.}, ``Blockchain-based peer-to-peer transactive energy trading (dt-p2pet),'' \emph{Energies}, vol.~15, no.~13, p. 4707, 2022. [Online]. Available: \url{https://pmc.ncbi.nlm.nih.gov/articles/PMC9269220/}
\BIBentrySTDinterwordspacing

\bibitem{MRV99}
\BIBentryALTinterwordspacing
S.~Micali, M.~Rabin, and S.~Vadhan, ``Verifiable random functions,'' in \emph{Proceedings of the 40th Annual Symposium on Foundations of Computer Science (FOCS)}.\hskip 1em plus 0.5em minus 0.4em\relax IEEE, 1999, pp. 120--130. [Online]. Available: \url{https://dl.acm.org/doi/10.5555/795665.796482}
\BIBentrySTDinterwordspacing

\bibitem{AlgorandSOSP2017}
\BIBentryALTinterwordspacing
Y.~Gilad, R.~Hemo, S.~Micali, G.~Vlachos, and N.~Zeldovich, ``Algorand: Scaling byzantine agreements for cryptocurrencies,'' in \emph{Proceedings of the 26th ACM Symposium on Operating Systems Principles (SOSP)}.\hskip 1em plus 0.5em minus 0.4em\relax ACM, 2017, pp. 51--68. [Online]. Available: \url{https://people.csail.mit.edu/nickolai/papers/gilad-algorand-eprint.pdf}
\BIBentrySTDinterwordspacing

\bibitem{HotStuff2019}
\BIBentryALTinterwordspacing
M.~Yin, D.~Malkhi, M.~K. Reiter, G.~G. Gueta, and I.~Abraham, ``Hotstuff: Bft consensus with linearity and responsiveness,'' in \emph{Proceedings of the 2019 ACM SIGSAC Conference on Computer and Communications Security (CCS)}.\hskip 1em plus 0.5em minus 0.4em\relax ACM, 2019, pp. 1--17. [Online]. Available: \url{https://dl.acm.org/doi/10.1145/3293611.3331591}
\BIBentrySTDinterwordspacing

\bibitem{Tendermint2018}
\BIBentryALTinterwordspacing
E.~Buchman, ``The latest gossip on bft consensus (tendermint),'' 2018. [Online]. Available: \url{https://arxiv.org/abs/1807.04938}
\BIBentrySTDinterwordspacing

\bibitem{Alshowkan2022SciRep}
\BIBentryALTinterwordspacing
M.~Alshowkan, P.~G. Evans, M.~Starke, D.~Earl, N.~A. Peters \emph{et~al.}, ``Authentication of smart grid communications using quantum key distribution,'' \emph{Scientific Reports}, vol.~12, no. 12731, 2022. [Online]. Available: \url{https://www.nature.com/articles/s41598-022-16090-w}
\BIBentrySTDinterwordspacing

\bibitem{Evans2021Access}
\BIBentryALTinterwordspacing
P.~G. Evans, G.~D. Peterson, T.~A. Morgan \emph{et~al.}, ``Trusted node qkd at an electrical utility,'' \emph{IEEE Access}, vol.~9, pp. 105\,220--105\,229, 2021. [Online]. Available: \url{https://doi.org/10.1109/ACCESS.2021.3070222}
\BIBentrySTDinterwordspacing

\bibitem{Tomamichel2011Leftover}
\BIBentryALTinterwordspacing
M.~Tomamichel, C.~Schaffner, A.~Smith, and R.~Renner, ``Leftover hashing against quantum side information,'' \emph{IEEE Transactions on Information Theory}, vol.~57, no.~8, pp. 5524--5535, 2011. [Online]. Available: \url{https://staff.science.uva.nl/c.schaffner/mypapers/TSSR11.pdf}
\BIBentrySTDinterwordspacing

\bibitem{NISTPQCproject}
{NIST}, ``Nist post-quantum cryptography standardization project,'' \url{https://csrc.nist.gov/projects/post-quantum-cryptography/post-quantum-cryptography-standardization}, accessed 2025-09-21.

\bibitem{NISTPQCEncryption}
{NIST}, ``Nist releases first 3 finalized post-quantum encryption standards,'' \url{https://www.nist.gov/news-events/news/2024/08/nist-releases-first-3-finalized-post-quantum-encryption-standards}, news release, Aug. 13, 2024.

\end{thebibliography}


\newpage
\vfill

\end{document}